\documentclass{eptcs}

\providecommand{\event}{RULE 2009} 
	\providecommand{\volume}{13}   
	\providecommand{\anno}{2009}   
	\providecommand{\firstpage}{13} 
	\providecommand{\eid}{2}       

\usepackage{alltt}
\usepackage[dvips]{graphicx}
\usepackage{amsmath}
\usepackage{amsthm}
\usepackage{macros,xspace}
\newcommand{\ito}{\Longrightarrow}

\newcommand{\rto}{\longrightarrow}
\newcommand{\Pin}{{\sc pin}\xspace}
\newtheorem{definition}{Definition}[section]

\newenvironment{program}{%
 \begin{quote}\begin{alltt}}
{\end{alltt}\end{quote}}

\title{An Implementation of Nested Pattern Matching
  in Interaction Nets}

\author{Abubakar Hassan
  \institute{Department of Computer Science \\
    University of Sussex \\ Falmer, Brighton, \\U.K.}
  \email{abubakar.hassan@sussex.ac.uk}  
  \and
  Eugen Jiresch
  \institute{Theory and Logic Group \\
  	Insitute of Computer Languages \\
  	Vienna University of Technology \\
  	Vienna, Austria}
  \email{jiresch@logic.at}  
  \and 
  Shinya Sato
  \institute{Faculty of Econoinformatics \\ Himeji Dokkyo University \\
    5-7-1 Kamiohno, Himeji-shi\\ Hyogo 670-8524, Japan}
  \email{shinya@himeji-du.ac.jp}  
}
\def\titlerunning{An Implementation of Nested Pattern Matching
  in Interaction Nets}
\def\authorrunning{Hassan, Jiresch \& Sato}
\begin{document}
\maketitle

\begin{abstract}
  Reduction rules in interaction nets
  are constrained to pattern match exactly one argument at a time. 
  Consequently, a programmer has to introduce  auxiliary rules
  to perform more sophisticated matches.
  In this paper, 
  we describe the design and implementation of
  a system for interaction nets which allows nested pattern matching
  on interaction rules. We achieve a system that provides convenient 
  ways to express interaction net programs without defining auxiliary
  rules.
  
\end{abstract}

\section{Introduction}
Interaction nets~\cite{Lafont90} were introduced over 10 years ago 
as a new programming paradigm based on graph rewriting.
Programs are expressed as graphs and computation is expressed
as graph transformation.
They enjoy nice properties such as locality of reduction,
strong confluence and Turing completeness. The definition
of interaction nets allows them to share computation: reducible
expressions (active pairs) cannot be duplicated. 
For these reasons, optimal and efficient $\lambda-$calculus 
evaluators~\cite{GonthierG:geoolr,LampingJ:algolc,MackieIC:yalyal}
based on interaction nets  have evolved.
Indeed, interaction nets have proved to be very useful for studying
the dynamics of computation. However, they remain fruitful only
for theoretical investigations. 

Despite that we can already program in
interaction nets, they still remain far from being 
used as a practical programming language. 
Drawing an analogy with functional programming,
we only have the $\lambda$-calculus without any high level language
constructs which provide programming comfort.
Interaction nets have a very primitive notion of 
pattern matching since only two agents can interact at a time.
Consequently, many auxiliary agents
and rules are needed to implement more sophisticated matches. 
These auxiliaries are implementation details and should be
generated automatically other than by the programmer. 

In this paper we take a step towards developing a richer 
language for interaction nets which facilitates nested pattern
matching. To illustrate what we are doing, consider the following
definition of a function that computes the last element of a list:
\begin{program}
lastElt (x:[]) = x
lastElt (x:xs) = lastElt xs
\end{program}
In this function {\tt []} is  a nested pattern in {\tt (x:[])}.
We cannot represent functions with nested patterns in interaction
nets. Hence, a programmer has to introduce auxiliary functions to 
pattern match the extra arguments.
\begin{program}
lastElt (x:xs) = aux x xs
aux x [] = x
aux x (y:ys) = lastElt (y:ys)
\end{program}

In our previous work~\cite{HaSa08} we defined a conservative extension of
interaction rules that allows nested pattern matching. The purpose of this
paper is to bring these ideas into practise:
\begin{itemize}
\item we define a programming language that captures the extended form 
of interaction rules;
\item we describe the implementation of these extended rules.
\end{itemize}
In~\cite{MaHaSa08} we defined a textual language for interaction nets (\Pin)
and an abstract machine that executes \Pin programs.
We take \Pin as our starting point and extend the \Pin language to allow
the representation of rules with nested patterns.

There has been several works that extend interaction nets in some way.
Sinot and Mackie's Macros for interaction nets~\cite{macros} are
quite close to what we present in this paper. They allow  pattern
matching  on more than one argument by relaxing the restriction of one 
principal port per agent. 
The main difference with our work is that their system does not
allow nested pattern matching. 
Our system facilitates nested/deep pattern matching
of agents. 

The rest of this paper is organised as follows:
In the Section~\ref{INs} we give a brief introduction of interaction nets.
In Section~\ref{inets} we define a programming language that allows
the definition of interaction rules with nested patterns.
In Section~\ref{implementation_overview} we give an overview of the
implementation of nested pattern matching. A more detailed explanation of the
algorithm is found in Section~\ref{wf_verification} (verification of
well-formedness) and Section~\ref{rule_translation} (rule translation). Finally,
we conclude the paper in Section~\ref{conclusion}.

\section{Interaction Nets} \label{INs}
We review the basic notions of interaction nets. 
See~\cite{Lafont90} for a more detailed presentation.
Interaction nets are specified by the following data:

\begin{itemize}
  
\item A set $\Sigma$ of \emph{symbols}. Elements of $\Sigma$ serve as
  \emph{agent} (node) labels. Each symbol has an associated arity $ar$
  that determines the number of its \emph{auxiliary ports}. If
  $ar(\alpha) = n$ for $\alpha \in \Sigma$, then $\alpha$ has $n+1$
  \emph{ports:} $n$ auxiliary ports and a distinguished one called the
  \emph{principal port}. 
  We represent an agent graphically as:
  \begin{net}{40}{45}
    \putagent{20}{20}{$\alpha$}
    \putDvector{20}{10}{10}
    \putline{12.6}{27.4}{-1}{1}{10}
    \putline{27.4}{27.4}{1}{1}{10}
    \puttext{20}{35}{$\cdots$}
    \put(2.6,38){\makebox(0,0)[br]{$x_1$}}
    \put(37.4,38){\makebox(0,0)[bl]{$x_n$}}
  \end{net}
  and textually using the syntax:
  $x_0 \sim \alpha(x_1, \ldots, x_n)$
  where $x_0$ is the principal port.

\item A \emph{net} built on $\Sigma$ is an undirected graph with agents
  at the vertices.  The edges of the net connect agents together at the
  ports such that there is only one edge at every port. 
  A port which is not connected is called a \emph{free port}. A set of free
  ports is called an \emph{interface}. A symbol denoting a free port is called
  a \emph{free variable}. 
  
\item Two agents $(\alpha,\beta)\in \Sigma\times\Sigma$
  connected via their principal ports form an
  \emph{active pair} (analogous to a redex).
  An interaction rule $((\alpha, \beta) \ito N) \in \IR$ replaces
  the pair $(\alpha, \beta)$ by the net $N$.  All the free ports are
  preserved during reduction, and there is at most one rule for each
  pair of agents. The following diagram illustrates the idea, where $N$
  is any net built from $\Sigma$.

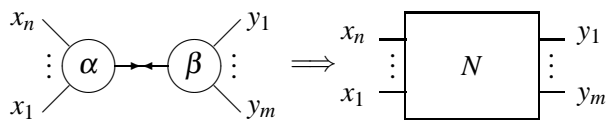
\begin{net}{200}{40}
  \putalpha{20}{20}
  \putbeta{60}{20}
  \putRvector{30}{20}{10}
  \putLvector{50}{20}{10}
  \putline{12.6}{27.4}{-1}{1}{10}
  \putline{12.6}{12.6}{-1}{-1}{10}
  \putline{67.4}{27.4}{1}{1}{10}
  \putline{67.4}{12.6}{1}{-1}{10}
  \puttext{5}{23}{$\vdots$}
  \puttext{75}{23}{$\vdots$}
  \put(0,0){\makebox(0,0)[br]{$x_1$}}
  \put(0,40){\makebox(0,0)[tr]{$x_n$}}
  \put(80,0){\makebox(0,0)[bl]{$y_m$}}
  \put(80,40){\makebox(0,0)[tl]{$y_1$}}
  \puttext{105}{20}{$\ito$}
  \putbox{140}{0}{50}{40}{$N$}
  \putHline{130}{10}{10}
  \putHline{130}{30}{10}
  \puttext{135}{23}{$\vdots$}
  \putHline{190}{10}{10}
  \putHline{190}{30}{10}
  \puttext{195}{23}{$\vdots$}
  \put(125,5){\makebox(0,0)[br]{$x_1$}}
  \put(125,35){\makebox(0,0)[tr]{$x_n$}}
  \put(205,5){\makebox(0,0)[bl]{$y_m$}}
  \put(205,35){\makebox(0,0)[tl]{$y_1$}}
\end{net}
For present purposes, we represent this rule textually using
$\left< \alpha(x_1,\ldots,x_n) \join \beta(y_1,\ldots,y_n)\right> \ito N$.
\end{itemize}

We use the notation $N_1 \rto N_2$ for the one step reduction
and $\rto^*$ for its transitive and reflexive closure.
Interaction Nets have the following property~\cite{Lafont90}:
\begin{itemize}
\item \textbf{Strong Confluence}: Let $N$ be a net. 
  If $N \rto N_1$ and $N \rto N_2$ with $N_1 \neq N_2$,
  then there is a net $N_3$ such that $N_1 \rto N_3$ and $N_2 \rto N_3$.
\end{itemize}

In Figure~\ref{rules_for_lastelm} we give a simple example of an interaction
net system that computes the last element of a list.
We  can represent  lists  using the agents Cons
(\agentt{:}) of arity 2 and $\nil$ of arity 0. 
The first port of Cons
connects to an element of the list and the second port of 
Cons connects to the
rest of the list. The agent $\nil$ marks the end of the list.
An active pair between \agentt{Lst} and Cons rewrites to an
auxiliary agent \agentt{Aux} with it's principal port oriented
towards the second auxiliary port of Cons. This means that during
computation, \agentt{Aux} will interact with either a Cons agent
or a $\nil$ agent\footnote{The second auxiliary port of a Cons agent 
will be connected to either a Cons agent or a Nil agent.}. To
avoid \emph{blocking} the computation, we define rules for active pairs
(\agentt{Aux}, $\nil$) and  (\agentt{Aux}, Cons).
An active pair between \agentt{Aux} and $\nil$ rewrites to a single wire,
which  connects the agents at the auxiliary ports of \agentt{Aux}.
When paired with Cons, \agentt{Aux} is replaced by \agentt{Lst},
analogous to the recursive call of the \texttt{lastElt} function. The list
element which is connected to the first port of \agentt{Aux} is deleted, as it is not the last element of the
list. This is modeled by the agent $\epsilon$, which erases all other agents.

\begin{figure}[htb]
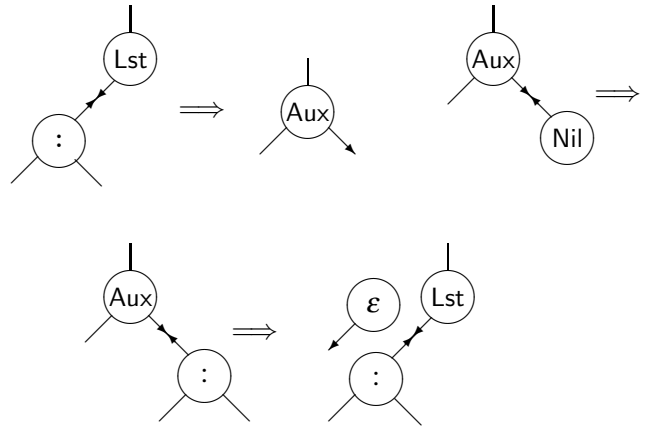

    \begin{net}{250}{165}
      \put(0,90){
        \begin{mnet}{102}{50}
          \putagent{43}{55}{\agentt{Lst}}
          \putagent{16}{24}{{\tt :}}
          \putline{8}{18}{-1}{-1}{10}
          \putline{22}{17}{1}{-1}{10}
          \putVline{43}{65}{10}
          \putline{37}{47}{-1}{-1}{15}
          \putvector{30}{40}{-1}{-1}{0}
          \putvector{30}{40}{1}{1}{0}         
          \puttext{70}{35}{$\ito$}
          \putagent{110}{35}{\agentt{Aux}}
          \putVline{110}{45}{10}
          \putline{102}{29}{-1}{-1}{10}
          \putvector{118}{29}{1}{-1}{10}
        \end{mnet}
      }
      \put(140,90){
        \putagent{43}{55}{\agentt{Aux}}
        \putVline{43}{65}{10}
        \putline{36}{48}{-1}{-1}{10}
        \puttext{90}{41}{$\ito$}
        \putVline{110}{10}{60}
        \putline{50}{48}{1}{-1}{15}
        \putvector{57}{41}{1}{-1}{0}
        \putvector{57}{41}{-1}{1}{0}
        \putagent{71}{25}{\nil}
      }
      \put(0,0){
        \begin{mnet}{102}{50}
          \put(0,0){
            \putagent{43}{55}{\agentt{Aux}}
            \putVline{43}{65}{10}
            \putline{36}{48}{-1}{-1}{10}
            \puttext{90}{41}{$\ito$}
            \putline{50}{48}{1}{-1}{15}
            \putvector{57}{41}{1}{-1}{0}
            \putvector{57}{41}{-1}{1}{0}
            \putagent{71}{25}{:}
            \putline{64}{18}{-1}{-1}{10}
            \putline{78}{18}{1}{-1}{10}
          }
          \put(120,0){
            \putagent{43}{55}{{\agentt{Lst}}}
            \putagent{16}{24}{:}
            \putline{8}{18}{-1}{-1}{10}
            \putline{23}{16.5}{1}{-1}{10}
            \putVline{43}{65}{10}
            \putline{37}{47}{-1}{-1}{15}
            \putvector{30}{40}{-1}{-1}{0}
            \putvector{30}{40}{1}{1}{0}
            \putagent{15}{52}{$\epsilon$}
            \putvector{8}{45}{-1}{-1}{10}
          }
        \end{mnet}
      }
    \end{net}
    \caption{Rules to compute last element of a list}
    \label{rules_for_lastelm}
\end{figure}
Figure~\ref{net_for_lastelm}  gives an example reduction sequence that
computes the last element of a list that contains just  one element: $[1]$. The
second port of \agentt{Lst} is free and thus acts as the interface of the net.
 \begin{figure}[htb]
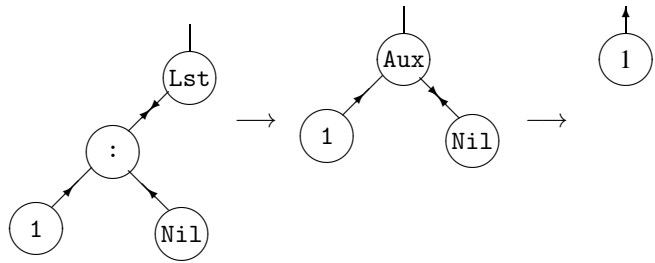

  \begin{small}
    \begin{net}{200}{80}
      \putagent{45}{52}{{\small {\tt Lst}}}
      \putagent{16}{24}{{\tt :}}
      \putline{8}{18}{-1}{-1}{15}
      \putline{22}{17}{1}{-1}{15}
      \putagent{-13}{-5}{{\tt 1}}
      \putagent{42}{-7}{{\tt Nil}}
      \putvector{1}{11}{1}{1}{0}      
      \putvector{29}{10}{-1}{1}{0}      
      \putVline{45}{62}{10}
      \putline{37}{47}{-1}{-1}{15}
      \putvector{30}{40}{-1}{-1}{0}
      \putvector{30}{40}{1}{1}{0}
      \puttext{70}{35}{$\rto$}
      \put(110,35){
        \putagent{16}{24}{{\tt Aux}}
        \putline{8}{18}{-1}{-1}{15}
        \putline{22}{17}{1}{-1}{15}
        \putagent{-13}{-5}{{\tt 1}}
        \putagent{42}{-7}{{\tt Nil}}
        \putvector{1}{11}{1}{1}{0}
        \putvector{29}{10}{-1}{1}{0}
        \putvector{29}{10}{1}{-1}{0}
        \putVline{16}{34}{10}
      }
      \puttext{180}{35}{$\rto$}
      \putagent{210}{59}{1}
      \putvector{210}{69}{0}{1}{10}
    \end{net}
    \caption{Example reduction sequence}
    \label{net_for_lastelm}
  \end{small}
\end{figure}
First, the active pair of \agentt{Lst} and Cons is rewritten,
introducing an \agentt{Aux} agent. Now the second rule is applied to the pair
$(\agentt{Aux} ,  \agentt{Nil})$, removing both agents and
connecting $\agentt{1}$ to the interface of the net.
As expected, this final net  contains only
the agent $\agentt{1}$, which is equivalent to the result of
\texttt{LastElt (1:[])}.

\paragraph*{Why Interaction Nets?}

Interaction nets are a generalisation of proof nets for linear
logic~\cite{GirardJY:linl}, in a similar way that term rewriting
systems are a generalisation of the $\lambda$-calculus.  Interaction
nets are an important model of computation for several reasons:

\begin{enumerate}
\item \emph{All} aspects of a computation are captured by the
rewriting rules---no external machinery such as copying a chunk of
memory, or a garbage collector, are needed. Interaction nets are
amongst the few formalisms which model computation where this is the
case, and consequently they can serve as both a low level operational
semantics and an object language for compilation, in addition to being
well suited as a basis for a high-level programming language.

\item Interaction nets naturally capture \emph{sharing}---active pairs
can never be duplicated. Thus only normal forms can be duplicated, and
this must be done incrementally. Using interaction nets as an
\emph{object} language for compilers has offered strong evidence that this
sharing will be passed on to the programming language being
implemented. One of the most spectacular instances of this is the work
by Gonthier, Abadi and L\'evy, who gave a system of interaction nets
to capture both optimal $\beta$-reduction~\cite{LevyJJ:phd} in the
$\lambda$-calculus~\cite{GonthierG:geoolr} (Lamping's
algorithm~\cite{LampingJ:algolc}), and optimal reduction for
cut-elimination in linear logic~\cite{GonthierG:linlwb}.

\item There is growing evidence that interaction nets can provide
a platform for the development of parallel implementations,
specifically parallel implementations of sequential programming
languages. Using interaction nets as an \emph{object} language for a
compiler offers strong evidence that the programming language being
implemented may be executed in parallel (\emph{even if there was no
parallelism in the original program}). Evidence of this can be found
in~\cite{PintoJ:phd}.

\item Finally, it is a formalism where both programs and
data structures are represented in the same framework. In the light of
the first point above, we have a very powerful mechanism to study the
dynamics of computation.
\end{enumerate}

The above points are clear indicators that interaction nets have a
role to play in computer science every bit as important as the
roles the $\lambda$-calculus or term rewriting systems have played
over the last few decades.  The aim of the present work is to realise
some of the potential of interaction nets through the development of 
practical programming constructs which simplify the task of
programming.

\subsection{Interaction rules with nested patterns - INP}
The above definition of interaction nets constraints pattern
matching to exactly one argument at a time. Consequently, we have
to introduce auxiliary agents and rules to perform deep pattern
matching (as exemplified in Figure~\ref{rules_for_lastelm}). 
Following~\cite{HaSa08}, an interaction rule may contain a \emph{nested
active pair} with more than two agents on it's left-hand side (lhs). 
A nested active pair is defined inductively as follows:
\begin{itemize}
\item Every active pair in ordinary interaction rules (ORN) is a nested
  active pair e.g. \\
$P=\left< \alpha(x_1,\ldots,x_n) \join \beta(y_1,\ldots,y_m)\right>$
\item A net obtained as a result of connecting the principal port
  of some agent $\gamma$ to a free port $y_j$ in a nested active pair $P$ is also
  a nested active pair e.g.
$\left< P, y_j \sim \gamma(z_1,\ldots,z_l)\right>$
\end{itemize}
As an example,
Figure~\ref{INP_lastelm} gives a  set
of INP rules that will compute the last element of a list.
In this Figure both rules contain a nested active pair on the lhs.
The (non interacting) agents {\tt Nil} and Cons 
on the lhs of the rules are nested agents.
These rules are compiled into the set of 
ORN rules given in Figure~\ref{rules_for_lastelm}
(See~\cite{HaSa08} for details of the compilation).

\begin{figure}[hth]
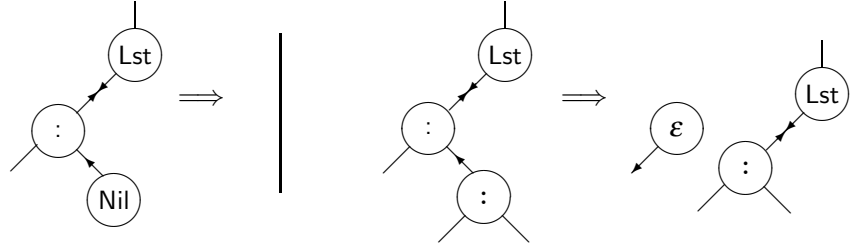

\begin{net}{320}{85}
    \put(0,15){
      \putagent{45}{52}{\agentt{Lst}}
      \putagent{16}{23}{\agentt{:}}
      \putline{8}{18}{-1}{-1}{10}
      \putline{23}{16}{1}{-1}{10}
      \putagent{37}{-3}{{\nil}}
      \putvector{26}{13}{-1}{1}{0}
      \putVline{45}{62}{10}
      \putline{38}{45}{-1}{-1}{15}
      \putvector{31}{38}{-1}{-1}{0}
      \putvector{31}{38}{1}{1}{0}
      \puttext{70}{35}{$\Lra$}
      \putVline{100}{0}{60}
    }
    \put(140,15){
      \putagent{45}{52}{\agentt{Lst}}
      \putagent{16}{24}{\agentt{:}}
      \putline{9}{17}{-1}{-1}{10}
      \putline{23}{17}{1}{-1}{10}
      \putagent{37}{-2}{{\tt :}}
      \putline{30}{-9}{-1}{-1}{10}
      \putline{44}{-9}{1}{-1}{10}
      \putvector{26}{14}{-1}{1}{0}
      \putVline{45}{62}{10}
      \putline{38}{45}{-1}{-1}{14}
      \putvector{31}{38}{-1}{-1}{0}
      \putvector{31}{38}{1}{1}{0}
      \puttext{75}{35}{$\Lra$}
      \putagent{110}{24}{$\epsilon$}
      \putvector{103}{17}{-1}{-1}{10}
      \put(120,-15){
        \putagent{45}{52}{{\agentt{\sf Lst}}}
        \putagent{16}{24}{{\tt :}}
        \putline{8}{18}{-1}{-1}{10}
        \putline{23}{17}{1}{-1}{10}
        \putVline{45}{62}{10}
        \putline{38}{45}{-1}{-1}{14}
        \putvector{31}{38}{-1}{-1}{0}
        \putvector{31}{38}{1}{1}{0}
}}
  \end{net}
  \caption{example INP rules to compute the last element of a list}
  \label{INP_lastelm}
\end{figure}

To ensure that INP preserves the Strong Confluence property of
interaction nets, rules in INP must satisfy the following
constraints~\cite{HaSa08}:

\begin{definition}[Sequentiality]\label{sequentiality}
Let $P$ be a nested active pair.
  The set of nested active pairs $\mathcal{P}$ is 
  sequential iff when 
$\left< P, y_j \sim \beta(x_1,\ldots,x_n)\right> \in \mathcal{P}$
  then
  \begin{enumerate}
  \item for the nested pair $P$, $P \in \mathcal{P}$ and, 
  \item for all the free ports $y$ in $P$ except the $y_j$ and for
    all agents $\alpha$, 
    $\left<P,y \sim \alpha(w_1,\ldots,w_n) \right> \notin \mathcal{P}$
  \end{enumerate}
\end{definition}
\begin{definition}[Well-formedness]\label{wfdef}
  A set of INP rules $\IR$ is well-formed iff
  \begin{enumerate}
  \item there is a sequential set of nested active pairs which contains every
  lhs of rules in $\IR$,
  \item for every rule $P \ito N$ in $\IR$, there is no interaction rule
    $P' \ito N'$ in $\IR$ such that $P'$ is a subnet of $P$.
  \end{enumerate}
\end{definition}
%
Intuitively, the Sequentiality property avoids overlaps between rules:
a set of INP rules containing nested patterns that violate
condition 2 of definition~\ref{sequentiality} can give rise to critical pairs,
which potentially destroys the Strong Confluence property of interaction nets.
Note that the definition of Sequentiality allows a nested active pair to
be a subnet of another nested active pair (in the same sequential set) which
may also give rise to critical pairs.
Definition~\ref{wfdef} ensures that there is at most one nested active pair
in any given set of rules.

\section{The Language}\label{inets}

We represent nets in the usual way as a comma
separated list of agents. This just corresponds to a flattening of the
net, and there are many different (equivalent) ways to do this
depending on the order in which the agents are enumerated. 
Using the net in Figure~\ref{net_for_lastelm} as an example
we write: 
\begin{program}
  p\(\sim\)Lst(r),p\(\sim\)Cons(x,xs),x\(\sim\)1,xs\(\sim\)Nil
\end{program}
The symbol `$\sim$' denotes the principal port of the agent. 
The variables \texttt{p},
\texttt{x} and \texttt{xs} are used to model the connection between two
ports.
All variable names occur at most twice: this limitation corresponds to
the requirement that it is not possible to have two edges connected to
the same port in the graphical setting. If a name occurs once,
then it corresponds to the free ports of the net ($r$ is free in the above).
If a name occurs twice, then it represents an edge between two ports. In this
latter case, we say that a variable is \emph{bound}. 

The syntax above can be simplified by replacing equals for equals:
\begin{program}
  Lst(r)\(\sim\)Cons(1,Nil)
\end{program}
In this
notation the general form of an active pair is 
$\alpha(\ldots)\sim \beta(\ldots)$.

We represent rules by writing $l \Lra r$, where $l$ is the net on the
left of the rule, and $r$ is the resulting net. In particular, we note
that $l$ will always consist of two agents connected at their
principal ports. As an example, 
the rules in Figure~\ref{rules_for_lastelm} are written as:

\begin{program}
   Lst(r) >< Cons(x,xs) \(\Lra\) xs\(\sim\) Aux(x,r)
   Aux(x,r) >< Nil \(\Lra\) x\(\sim\)r
   Aux(x,r) >< Cons(y,ys) \(\Lra\) Lst(r)\(\sim\)Cons(y,ys), \(\epsilon\sim \)x
\end{program}
The names of the bound variables in the two nets must be disjoint, and
the free variables must coincide, which corresponds to the condition
that the free variables must be preserved under reduction.  
Note that we use the symbol `\texttt{><}' for the active pair of the rule
so that we can distinguish between an active pair and a rule.

We can simplify only \emph{optimised} ORN rules i.e. rules that do not contain
active pairs on their right-hand side (rhs).
For example we can write the set 
of rules above in a simplified form:
\begin{program}
   Lst(r) >< Cons(x,Aux(x,r))
   Aux(x,x) >< Nil
   Aux(x,r) >< Cons(y,ys) \(\Lra\) Lst(r)\(\sim\)Cons(y,ys), \(\epsilon\sim \)x
\end{program}
Note that once an optimised ORN rule is simplified, it's rhs becomes empty 
and therefore we omit the `$\Lra$'.

We represent INP rules using a similar mechanism and allow the lhs
of a rule to contain more than two agents. We restrict the simplification
of INP rules such that the lhs and rhs are simplified independently of
one another.
As an example, the set of INP rules in Figure~\ref{INP_lastelm} can be 
written as:
\begin{program}
   Lst(r) >< Cons(x,xs), xs\(\sim\)Nil \(\Lra\) r\(\sim\)x
   Lst(r) >< Cons(x,xs),xs\(\sim\)Cons(y,ys) \(\Lra\) x\(\sim\epsilon\),p\(\sim\)Lst(r),p\(\sim\)Cons(y,ys)
\end{program}
or in a simplified form:
\begin{program}
   Lst(r) >< Cons(x,Nil)\(\Lra\) r\(\sim\)x
   Lst(r) >< Cons(x,Cons(y,ys))\(\Lra\) x\(\sim\epsilon\),Lst(r)\(\sim\)Cons(y,ys)
\end{program}

\paragraph{Other language constructs}
The \Pin system provides other language constructs --
modules,  built-in operations
on agent values, input/output e.t.c. These constructs 
remain unaffected by our extension and are out of scope
in this paper. See~\cite{MaHaSa08}
for a detailed description of the additional language features.

\section{Translation}\label{implementation_overview}
In this Section we give an overview of our translation. 
In general, the \Pin compiler reads programs in our source language
and builds the corresponding abstract syntax tree (AST).
On the basis of the AST, the \Pin compiler generates some byte codes 
which can be executed by an abstract machine or be further compiled into
C source code. See~\cite{MaHaSa08} for  a more detailed presentation
of the \Pin system.

Our translation function rewrites ASTs that represent INP rules
into ASTs that represent ordinary interaction rules.
Therefore, the back end of the \Pin system remains unaffected 
by the translation.
Overall, our translation function is similar to the
compilation schemes defined in the original paper \cite{HaSa08}.
We  summarise the translation algorithm in the following
steps:



\begin{enumerate}
\item A rule is found in the AST. This rule can be either ORN or INP. All
  other nodes of the AST that are not rules (imports, variable
  declarations,\ldots) are ignored.
\item Check if the lhs is not a subnet of a previously translated
  lhs (and vice versa if the rule is INP). 
  This is the first part of verifying the well-formedness
  property (Definition~\ref{wfdef}). We discuss this verification in
  Section~\ref{subnetprop}.
\item If the rule is not INP, return.
\item If the rule is INP: check if the current and all previous 
  nested active  pairs can be added to a sequential set. 
  This is the second part of verifying
  the well-formedness property (Section~\ref{sequentialprop}).
\item If both checks are passed, translate the rule (else, exit with an
  error message):
  \begin{enumerate}
  \item Resolve the first nested agent of the rule's active pair.
  \item Add an auxiliary rule to the AST.
  \item The remaining nested agents are not (yet) translated. 
    They are resolved by translating the auxiliary rule.
  \end{enumerate}
  We describe the translation algorithm in Section~\ref{rule_translation}.
\item traverse the AST until the next (unprocessed) rule is found.
\end{enumerate}
This algorithm allows for an arbitrary number of 
nested patterns (i.e., the number of nested agents in the lhs of an INP rule) 
and an arbitrary pattern depth.

\section{Verifying the Well-Formedness Property}
\label{wf_verification}
Our verification algorithm (see below) consists of two parts
which correspond to the two constraints of the well-formedness
property. We verify that the set of nested active pairs in  a given
\Pin program are both disjoint and sequential.

We use the notation $\verb|[]|$ for the empty list, 
$\verb|[| 1,\ldots,n\verb|]|$ for
a list of $n$ elements and $ps_1\mathtt{@}ps_2$ to append two lists.

\newcommand{\PS}{\mathrm{PositionSet}}
\newcommand{\Pos}{\mathrm{Pos}}
\newcommand{\PosSym}{\mathrm{PosSym}}


%
%
%
%

\begin{definition}[Position Set]
Let $l \Lra r$ be a rule in \Pin. 
We define the function $\PosSym(l)$ that given a nested active pair
will return  a set of 
pairs $(ps,u)$ where $ps$ is a list that represents the position
of a symbol $u$ in $l$.
\[
\begin{array}{lcl}  
  \PosSym(\alpha \verb|(| t_1 \verb|,| \ldots \verb|,| t_n \verb|)| 
  \ \verb|><| \
  \beta\verb|(| s_1 \verb|,| \ldots \verb|,| s_m \verb|)|)
  & = &
  \PosSym_t(\verb|[| 1,1 \verb|]|, \ t_1) \; \cup \, ... \, \cup 
   \PosSym_t(\verb|[| 1,n \verb|]|, \ t_n) \; \cup\\
   &&\PosSym_t(\verb|[| 2,1 \verb|]|, \ s_1) \; \cup \, \ldots \, \cup \;
   \PosSym_t(\verb|[| 2,m \verb|]|, \ s_m)\\
   \PosSym_t(\mathit{ps}, \ x) &=& \emptyset\\

   \PosSym_t(\mathit{ps}, \ \alpha\verb|(| \bar{x} \verb|)|) 
   & = & 
   \{(\mathit{ps}, \alpha)\} \\
   %
   %
   \PosSym_t(\mathit{ps}, \ \alpha\verb|(| t_1\verb|,| \ldots
   \verb|,| t_n \verb|)|) 
   & = & 
   \PosSym_t(\mathit{ps} \, \verb|@| \, \verb|[| 1 \verb|]|, \ t_1) \;
   \cup \, ... \, \cup \;
   \PosSym_t(\mathit{ps} \, \verb|@| \, \verb|[| n \verb|]|, \ t_n)\\
   %
   %
\end{array}
\]
where the
sequence of terms $t_1,\ldots,t_n$ in 
$\PosSym_t(\mathit{ps}, \ \alpha\verb|(| t_1\verb|,| \ldots
 \verb|,| t_n \verb|)|)$ contain at least one term which is not a variable;
and $\bar{x}$ is a sequence of zero or more variables.



%

%
%
%

\noindent {}The function $\Pos(l)$ returns a set of lists
that represent the position of each nested agent in a nested active
pair:
\[
\begin{array}{lcl}
\Pos(l) & = & \pi_1(\PosSym(l))\\
\pi_1(\emptyset) & = & \emptyset,\\
\pi_1(\{(ps,s)\} \, \cup \, A) & = & \{ps\} \, \cup \, \pi_1(A-\{(ps,s)\}).\\
\end{array}
\]
\noindent{}We extend these operations into the sequence $l_1,...,l_k$
of  lhs of rules  as follows:
\[
\begin{array}{lcl}
  \PosSym(l_1,\ldots,l_k) 
  & = &
  \PosSym(l_1) \; \cup \, \ldots \, \cup \; \PosSym(l_k),\\

  \Pos(l_1,\ldots,l_k) &= &\Pos(l_1) \; \cup \, ... \, \cup \; \Pos(l_k)\\
\end{array}
\]

\end{definition}

\newtheorem{example_positionset}[definition]{Example}
\begin{example_positionset}
For each rule in Figure~\ref{INP_lastelm}, we can get sets of positions of
nested agent pairs as follows:
\begin{itemize}
\item $\PosSym(\verb|Lst(r) >< Cons(x,Nil)|)$

\quad $= \ \PosSym_t(\verb|[| 1,1 \verb|]|, \verb|r|) \, \cup \, 
 \PosSym_t(\verb|[| 2,1 \verb|]|, \verb|x|) \, \cup \, 
 \PosSym_t(\verb|[| 2,2 \verb|]|, \verb|Nil|) \ = \ 
 \{ (\verb|[| 2,2 \verb|]|, \verb|Nil|) \}$

\item $\Pos(\verb|Lst(r) >< Cons(x,Nil)|) \ = \ \{ \verb|[| 2,2 \verb|]| \}$

\item $\PosSym(\verb|Lst(r) >< Cons(x,Cons(y,ys))|)$

\quad $= \ \PosSym_t(\verb|[| 1,1 \verb|]|, \verb|r|) \, \cup \, 
 \PosSym_t(\verb|[| 2,1 \verb|]|, \verb|y|) \, \cup \, 
 \PosSym_t(\verb|[| 2,2 \verb|]|, \verb|Cons(y,ys)|) \ = \ 
 \{ (\verb|[| 2,2 \verb|]|, \verb|Cons|) \}$

\item $\Pos(\verb|Lst(r) >< Cons(x,Cons(y,ys))|) \ = \ \{ \verb|[| 2,2 \verb|]| \}$

\end{itemize}
\end{example_positionset}
\subsection{Subnet property} \label{subnetprop}
Verifying the subnet property is straightforward. 
Since rules are represented as trees (subtrees of the AST), 
it is easy to verify if one rule's lhs is a subtree of another.
We compute the lhs subtree relation of the current rule $P$ (to
be translated) against all the rules $Q$ which have already been translated.
If $P$
is in ORN, we verify the subtree relation in only one direction: 
the lhs of an INP rule cannot be a subnet of the lhs of an ORN rule.
Otherwise we verify the subtree relation in both directions:
$P$ against $Q$ and $Q$ against $P$.
The case of two ORN
rules with the same active agents is handled by the compiler at an 
earlier stage.
If the current rule's lhs is not a subnet of any previous rules' 
lhs's, we add it to the set of previous rules.

Note that we consider a tree to be a subtree of another tree up to alpha
conversion, i.e., variable names are not considered.

\subsection{Sequential set property} \label{sequentialprop}
The check for the sequential set property is a bit more complicated 
than the subnet one. According to the definition of the 
well-formedness property, 
there must exist a sequential set that 
contains all nested active pairs in a given set of INP
rules. 
Rather than attempting to construct such a sequential set, 
the algorithm tries to falsify this condition: 
it searches (exhaustively) for two nested patterns
that cannot be in the same sequential set. This is done as follows:

For the current nested pattern $P$ and all previously verified 
patterns $Q$ with the same active agents:
\begin{enumerate}
\item Compare the sets $\Pos(P)$ and $\Pos(Q)$.
  We only consider the positions of agents at this point, not the agents
  themselves.
\item If one set is a subset of another, $P$ and $Q$ can be added to a
  sequential set 
  \footnote{Note that $P$ and $Q$ have already passed the subnet check at
    this point. This means that (some of) the nested agents at the common
    positions are different. Hence, $P$ and $Q$ cannot give rise to 
    a critical
    pair.}. $P$ is added to the set of previous nested patterns.
\item Else, we compare the actual nested agents at the common positions
$CP=\Pos(P) \cap \Pos(Q)$.

\item If for each element $p \in CP$, $\alpha$ and $\beta$ are the same
where $(p,\alpha) \in \PosSym(P)$ and $(p,\beta) \in \PosSym(Q)$,
no sequential set can contain $P$ and $Q$, as $P \equiv \langle M,
x\sim\alpha(\ldots),\textbf{a} \rangle, Q \equiv \langle M,
  y\sim\beta(\ldots) ,\textbf{b}\rangle$ with $x \neq y$

\item Else, $P$ and $Q$ can be added to a sequential set. $P$ is added to the
  set of previously verified nested patterns.
\end{enumerate}

\paragraph*{}
It is straightforward to see that after the full traversal of the AST, 
all possible pairs
of nested patterns are considered. Hence, the search for a pair that violates
the sequential set property is exhaustive.

\theoremstyle{theorem}
\newtheorem{wfcorrect}[definition]{Proposition}
\begin{wfcorrect}
  Let $R$ be a set of INP rules.  
  $R$ is well-formed $\Leftrightarrow$ $R$ is correctly verified to be
well-formed using our verification algorithm.
\end{wfcorrect}

\begin{proof}[Proof]
  
  $\Leftarrow:$
  
  Assume $R$ is not well-formed and passes the well-formedness checks. 
  We proceed by a complete case distinction (according to the
  definition of well-formedness):
  
  \begin{description}
  \item[Case 1.] 
    There exist two rules $ P\ito N, Q \ito M \in R $ where $P$ is
    a subnet of $Q$. But then, the pair $(P,Q)$ is tested for 
    the subnet relation (Section~\ref{subnetprop}). 
    Hence, $R$ does not pass the 
    well-formedness check.
  \item[Case 2.] 
    There exist two rules $A\ito N,B \ito M \in R$ where 
    $A \equiv \langle P, x\sim\alpha(\ldots) \rangle, 
    B \equiv \langle P, y\sim\beta(\ldots)
    \rangle$ for $x \neq y$. 
    But since all pairs of nested patterns are checked
    for the sequential set property (Section~\ref{sequentialprop}), 
    $(A,B)$ will be detected. 
    Hence, $R$ does not pass the check.
  \end{description}
  
  In both cases, we reach a contradiction to the assumption above, 
  hence it cannot be true. \\

  $\Rightarrow:$
  
  Assume $R$ does not pass the well-formedness check, 
  but is well-formed. Again, there are only two cases:
  
  \begin{description}
  \item[Case 1.] 
    R does not pass the check because $\exists P\ito N, Q \ito N'
    \in R $ where $P$ is a subnet of $Q$. 
    But then, $R$ is not well-formed (by the definition of well-formedness).
    
  \item[Case 2.] 
    R does not pass because are two rules $A \ito N,B 
    \ito M \in R$ where 
    $A \equiv \langle P, x\sim\alpha(\ldots) \rangle, B \equiv
    \langle P, y\sim\beta(\ldots) \rangle$ for $x \neq y$. 
    Then, there is no sequential set that contains both
    $A$ and $B$. Hence, $R$ is not well-formed.
  \end{description}

  Again, we reach a contradiction to the assumption above in either case.
  
\end{proof}

\section{Rule Translation}\label{rule_translation}
We now describe the translation algorithm in more detail. 
As mentioned earlier, we translate INP rules to ORN rules by
rewriting the AST.
We perform a pre-order traversal of the AST and identify nodes
that represent INP rules.
Once we find an INP rule, we replace its nested agents with a fresh
(variable) node $n$ and replace the subtree that represents the
rhs of the rule with a new tree $N_t$. 
The nested agents and the rhs of the rule are 
stored for later processing. The tree $N_t$ represents an active pair
between $n$ and a newly created auxiliary agent $Aux$.
This auxiliary agent holds all the agents and attributes of the 
original active pair, with the exception of the former variable agent.

Now, we create an auxiliary rule with an active pair between
$Aux$ and the current nested agent
(initially connected to the interacting agent).
We set the rhs of the auxiliary rule to be the rhs of the original INP rule.
Finally, we add this auxiliary rule to the system.

Note that the auxiliary agent in the new rule may still contain additional
nested agents, i.e., the auxiliary rule may be INP. 
Hence,  the translation algorithm recursively translates the generated rules
until the lhs of each of the generated rules contains exactly two agents.
The idea behind this is to resolve one nested agent per translation pass.
Further nested agents are processed when the translation function reaches the
respective auxiliary rule(s).

We can formalise the translation algorithm as a function
$translate(\mathcal{R},U,S)$, where $\mathcal{R}$
denotes the input set of interaction rules and $U$ and $S$ are \emph{stores}.
Intuitively, the components $U$ and $S$ are used to store
previously processed rule patterns in order to verify the subnet and 
sequential set properties respectively. 
The function $translate$ is defined as follows (in pseudo-code
notation), where \texttt{FAIL} denotes termination of the program due to non
well-formedness of $\mathcal{R}$:

\begin{program}
  translate([],\(U\),\(S\)) = [] 
  translate((P\(\ito\)N):\(\mathcal{R}\),\(U\),\(S\)) = 
  if (P is a subnet of any Q \(\in U\) or vice versa) 
  \indent FAIL 
  else if (P is ORN) 
  \indent (P\(\ito\)N):translate(\(\mathcal{R}\),P:\(U\),\(S\)) 
  else 
  \indent if (P cannot be added to a sequential set with 
  any Q \(\in S\))
  \indent \indent FAIL
  \indent else
  \indent (P'\(\ito\)(PX\(\sim\)p)):translate((PX >< A \(\ito\) N):\(\mathcal{R}\),P:\(U\),P:\(S\)) 
  \indent \indent where
  \indent \indent \indent p = position of the first nested agent of P
  \indent \indent \indent A = the nested agents at position p
  \indent \indent \indent P' = P with all nested agents replaced by variables
  \indent \indent \indent \indent \indent \indent \indent \indent ports
  \indent \indent \indent PX = auxiliary agent that contains all ports of P
  \indent \indent \indent \indent \indent \indent \indent \indent except p
\end{program}

\theoremstyle{definition}
\newtheorem{translation_example}[seqdef]{Example}
\begin{translation_example}\label{translation_example}
Consider the interaction rules from Figure~\ref{INP_lastelm}. $\mathcal{R}$
consists of two rules:
\begin{program}
   1. Lst(r) >< Cons(x,Nil) \(\Lra\) r\(\sim\)x
   2. Lst(r) >< Cons(x,Cons(y,ys)) \(\Lra\) x\(\sim\epsilon\),Lst(r)\(\sim\)Cons(y,ys)
\end{program}
The translation works as follows:
\begin{itemize}
  \item Rule 1 is INP (due to the {\tt Nil} agent).\\
  Its lhs is checked for the subnet property. As there are no previous
  rules in $U$, the check is passed.
  \item As the rule is INP, it is checked for the sequential set property.
  Again, there are no previous rules in $S$, hence it passes the check.
  \item The rule is transformed, introducing the auxiliary
  agent {\tt Lst\_Cons}

  		\begin{program}
         1. Lst(r) >< Cons(x,var0) \(\Lra\) Lst_Cons(r,x)\(\sim\)var0
      \end{program}
  \item A new auxiliary rule is added to $\mathcal{R}$:
  \begin{program}
         3. Lst_Cons(r,x) >< Nil \(\Lra\) r\(\sim\)x
      \end{program}
		The lhs pattern of rule 1 is added to $U$ and $S$.
  \item Rule 2 is INP (due to the nested {\tt Cons} agents).\\
  Its lhs is checked for the subnet property. The only lhs
  pattern in $U$ is {\tt Lst(r) >< Cons(x,Nil)}. Due to different agents at the
  second auxiliary port of {\tt Cons}, the lhses cannot be subnets of one another. The
  check is passed.
  \item Rule 2 is checked for the sequential set property. First, the positions
  of nested agents of Rule 1 and 2 are compared. Since they are the same (both
  have their nested agent at the second auxiliary port of {\tt Cons}), they can
  be added to a sequential set. There are no further rules in $S$, hence the
  check is passed.
  \item The rule is transformed and another auxiliary rule is added to
  $\mathcal{R}$:
  \begin{program}
   2. Lst(r) >< Cons(x,var1) \(\Lra\) Lst_Cons(r,x)\(\sim\)var1
   4. Lst_Cons(r,x) >< Cons(y,ys) \(\Lra\)  x\(\sim\epsilon\),Lst(r)\(\sim\)Cons(y,ys)
	\end{program}
  \item Rule 3 is ORN.\\
  Its lhs is checked for the subnet property. As there are no with the same
  active agents in $U$, the check is passed. The lhs pattern of Rule 3 is added
  to $U$.
  \item Rule 4 is ORN.\\
  Its lhs is checked for the subnet property. Again, there are no rules with
  the same active agents in $S$. The check is passed and the lhs pattern is
  added to $U$.
\end{itemize}
This yields the translated set of rules
\begin{program}
         1. Lst(r) >< Cons(x,var0) \(\Lra\) Lst_Cons(r,x)\(\sim\)var0
         2. Lst(r) >< Cons(x,var1) \(\Lra\) Lst_Cons(r,x)\(\sim\)var1
         3. Lst_Cons(r,x) >< Nil \(\Lra\) r\(\sim\)x
         4. Lst_Cons(r,x) >< Cons(y,ys) \(\Lra\)  x\(\sim\epsilon\),Lst(r)\(\sim\)Cons(y,ys)
\end{program}
Note that the rules 1 and 2 are identical (save variable names). Since only one
of them is needed, rule 2 is discarded by \Pin.
As expected, we get the set of ORN rules given  in Figure~\ref{rules_for_lastelm}.
\end{translation_example}

\newtheorem{translate_terminates}[seqdef]{Proposition}
\begin{translate_terminates}
\textbf{(Termination) } For a finite $\mathcal{R}$, $translate$ terminates.
\begin{proof}
	Let $n$ be the number of rules in $\mathcal{R}$ and $p$ be the sum of all
	nested agents of these rules. By a complete case distinction, we show that with
	each recursive call, $(n+p)$ decreases:
	\begin{description}
	\item[Case 1] The current rule is ORN. At the recursive call, 
          it is removed
	from $\mathcal{R}$, hence $n$ decreases by 1.
	\item[Case 2] The current rule has $i$ nested agents ($i>0$). 
          The rule is
	removed from $\mathcal{R}$ and an auxiliary rule with 
        exactly $i-1$ nested
	agents is added to $\mathcal{R}$. Hence, with the recursive 
        call $p$ decreases by 1.
	\end{description}
$translate$ terminates if $n=0$. $n$ only decreases if we encounter an ORN
rule. Yet, since the number of nested agents for each rule is finite, all rules
in $\mathcal{R}$ will be ORN after finitely many decreases of $p$.
\end{proof}
\end{translate_terminates}

\subsection{Time and space complexity of
$translate$}\label{translate_complexity}
\paragraph{Time complexity}
To determine the time complexity of $translate$, we analyze the three main
parts of the algorithm: the check for the subnet property, the check for the
sequential set property and the actual rule translation. As a measurement of
the input size, we consider the sum $n = r + a$, where $r$ is the number of
rules and $a$ is the number of nested agents in the input set
$\mathcal{R}$. The idea behind this notion of size is that $translate$ is invoked 
once for every interaction rule in the input. Additionally, each
elimination of a nested agent yields a new auxiliary rule to be translated.
Hence, $translate$ is called $n=r+a$ times.

The subnet property check compares all pairs of rule patterns. For an
input set of size $n$, $\frac{n(n-1)}{2}$ checks are performed. Therefore, the
verification of the subnet property is quadratic to the input size, or $O(n^2)$.

At first glance, the check for the sequential set property is similar to the
one for the subnet property: all pairs of rule patterns need to be compared,
resulting in a (worst-case) complexity of $O(n^2)$. However, the verification
algorithm only considers pairs of patterns with the same active agents (see
\ref{sequentialprop}), which is usually only a fraction of all pairs of rule
patterns. On average, this results in a less than quadratic complexity for the 
verification of the sequential set
property.

The actual translation is executed once for each nested agent in the input set.
The elimination of a single nested agent is not influenced by the total
number of nested agents or rules. The rule translation is thus linear to the
number of nested agents $r$.

This results in an overall time complexity of $O(n^2)$ for $translate$, which is
mostly determined by the verification of the subnet property. Extended
profiling on the implementation of the algorithm (including input sets
with up to several thousand nested agents and rules) shows that also on
practical cases the performance is bounded by a quadratic curve (with small constants involved).

\paragraph{Space complexity }
To analyze the space complexity of $translate$, we consider the population of
the stores $U$ and $S$. Every rule pattern of the input set is stored in $U$,
whereas only nested patterns are stored in $S$. At the worst case, $U$ and $S$
will contain $2n$ rule patterns. This implies that space complexity is linear
to the number of rules in the input. Again, this result is reflected by
several profiling tests using practical examples.

\subsection{Additional language features}\label{language_features}

The \Pin language offers some features that are not considered in the
original definition of the nested pattern translation function. Some important
examples are data values of agents (integers, floats, strings,\ldots),
side effects (declaration and manipulation of variables, I/O) and conditions.
These features are not involved in the process of nested
pattern matching. Therefore, they do not need to be processed or changed by the
translation function:
\begin{itemize}
\item with regard to nested pattern matching, data values can be considered as
  variable ports (they do not contain nested agents). Hence, they are unaffected
  by the translation.
\item conditionals and side effects only occur in the rhs of a rule. Since the
  original rhs of an INP rule is propagated to the final auxiliary rule without a
  change, these features are not affected either.
\item all auxiliary rules but the last one are responsible for pattern
  matching only, they do not do the ``actual work'' of the original rule. All
  special language features are simply passed to the next auxiliary rule.
\end{itemize}

\subsection{The implementation}
We have developed a prototype implementation which
can be obtained from the project's web page
\url{http://www.interaction-nets.org/}.
We have thoroughly tested the prototype implementation and 
developed several example modules. These examples include 
rule systems with a large number and depth of 
nested patterns as well as heavy use of state, conditionals
and I/O. Additionally, we have designed several non well-formed systems 
in order to
improve error handling. Some of these examples can be found in the current
\Pin distribution at the project's web page.

\section{Conclusion}\label{conclusion}
We have presented an implementation for nested pattern matching of interaction
rules. The implementation closely follows the definition of nested patterns and
their translation to ordinary patterns introduced in ~\cite{HaSa08}.
We have shown nice properties of the algorithm such as its correctness and
termination.

The resulting system allows programs to be expressed in a more convenient way
rather than introducing auxiliary agents and rules to pattern match nested 
agents. We see this as a positive step
for further extensions to interaction nets: future
implementations of high-level language constructs can be built upon these more
expressive rules.



\bibliographystyle{plain} 
\bibliography{bibfile}

\end{document}